\documentclass{article}
\usepackage{graphicx}% Include figure files
\usepackage{dcolumn}% Align table columns on decimal point
\usepackage{bm}% bold math
\usepackage{color}
\usepackage[utf8]{inputenc}
\usepackage[T1]{fontenc}
\usepackage{mathptmx}
\usepackage{amsthm}
\usepackage{amsmath}
\usepackage{amssymb}

\usepackage{authblk}
\usepackage{amsmath}
\usepackage{amsfonts}
\usepackage{graphicx,subcaption}
\usepackage{url}
\newtheorem{theorem}{Theorem}[section]

\begin{document}

\title{Spatiotemporal characteristics in systems of diffusively coupled excitable slow-fast FitzHugh-Rinzel dynamical neurons}
% Force line breaks with \\
\author[1]{Arnab Mondal}%

\author[2,3]{Argha Mondal \thanks{Corresponding author: arghamondalb1@gmail.com}}

\author[1]{Sanjeev Kumar Sharma}%

\author[1]{Ranjit Kumar Upadhyay}%

\author[3]{Chris G. Antonopoulos}
%\email{http://www.Second.institution.edu/~Charlie.Author.}

\affil[1]{Department of Mathematics and Computing, Indian Institute of Technology (Indian School of Mines), Dhanbad 826004, India}
\affil[2]{School of Engineering, Amrita Vishwa Vidyapeetham, Amritapuri, Kollam 690525, India}
\affil[3]{Department of Mathematical Sciences, University of Essex, Wivenhoe Park, Colchester CO4 3SQ, UK}

\date{\today}% It is always \today, today,
 % but any date may be explicitly specified

\maketitle

\begin{abstract}
In this paper, we study an excitable, biophysical system that supports wave propagation of nerve impulses. We consider a slow-fast, FitzHugh-Rinzel neuron model where only the membrane voltage interacts diffusively, giving rise to the formation of spatiotemporal patterns. We focus on local, nonlinear excitations and diverse neural responses in an excitable 1- and 2-dimensional configuration of diffusively coupled FitzHugh-Rinzel neurons. The study of the emerging spatiotemporal patterns is essential in understanding the working mechanism in different brain areas. We derive analytically the coefficients of the amplitude equations in the vicinity of Hopf bifurcations and characterize various patterns, including spirals exhibiting complex geometric substructures. Further, we derive analytically the condition for the development of antispirals in the neighborhood of the bifurcation point. The emergence of broken target waves can be observed to form spiral-like profiles. The spatial dynamics of the excitable system exhibits 2- and multi-arm spirals for small diffusive couplings. Our results reveal a multitude of neural excitabilities and possible conditions for the emergence of spiral-wave formation. Finally, we show that the coupled excitable systems with different firing characteristics, participate in a collective behavior that may contribute significantly to irregular neural dynamics.
\end{abstract}

\vskip 1cm

\noindent\textbf{
In this paper, we study the slow-fast FitzHugh-Rinzel biophysical excitable neuron model. We demonstrate the appearance of diverse dynamical behavior depending on spike-bursting, firing responses. We also discuss analytically and numerically the bifurcation analysis. To understand the effects of diffusion on spatiotemporal pattern formation in coupled FitzHugh-Rinzel neurons, we consider a FitzHugh-Rinzel system where only the membrane voltage interacts diffusively. We compute the synchronization index to study the collective behavior of the coupled systems and describe various patterns, including target waves, multi-arm spirals and antispirals. Finally, we derive analytically the amplitude equations to validate the existence of antispiral patterns in the vicinity of Hopf bifurcations.}

\section{Introduction}

The complex dynamics of excitable neurons is related to ionic concentrations inside and outside neural membranes and play a major role in generating action potentials. This provides a better way to understand brain functioning in normal and pathophysiological states \cite{ermentrout1998neural,izhikevich2000neural,izhikevich2007dynamical,keane2015propagating,ma2017review,townsend2018detection}. Neurons receive incoming sensory inputs, encode them into different biophysical variables and produce relevant outputs \cite{izhikevich2007dynamical}. Biophysical mechanisms are dynamical in nature and can be studied using concepts from dynamical systems theory. In extended biophysical systems of coupled excitable cells, neural populations produce a transmembrane potential difference that travels across neurons by means of wave propagation \cite{kondo2010reaction,meier2015bursting,milton1993spiral,townsend2018detection}. This may cause the appearance of synchronous activity in groups of neighboring neurons. Such types of synchronized behavior are playing a major role in signal processing in neural populations \cite{belykh2005synchronisation}.

Here, we consider a slow-fast, FitzHugh-Rinzel (FHR) neuron model where only the membrane voltage interacts diffusively, giving rise to the formation of spatiotemporal patterns. The FHR model exhibits diverse bursting activities. In this context, bursting refers to rapid changes in membrane voltage oscillations which involve simultaneous changes between active phases (high amplitude oscillations or bursts of spikes) and silent phases (low amplitude/subthreshold oscillations or amplitude death, i.e., quiescence) \cite{izhikevich2000neural}. We study a dynamical system of three coupled partial differential equations (PDEs) for 2-dimensional spatial pattern formation based on the FHR model, where the recovery and/or slow variables (slow-fast dynamics) interact with the membrane voltage variable. We seek to explore spatial instabilities considering 1- and 2-dimensional diffusion processes and the nature of spatiotemporal patterns. It is the symmetry breaking nature of spatial systems that gives rise to Turing-like structures \cite{turing1952philosophical}. Only one spatial variable, i.e., the membrane potential, determines the scale of the patterns observed. In our system, there exists only one spatially distributed scale, and thus the spatial instability occurs due to the dynamics rather than due to differences in spatial scales \cite{mondal2019diffusion}.

The mathematical analysis for the emergence of spatial structures is important to understand a wide range of biophysical and pathological phenomena \cite{kondo2002reaction,erichsen2008multistability,grace2015reg,kilpatrick2010effects,roxin2004self,song2018classification}. The work here is motivated by our earlier work in \cite{mondal2019diffusion} and other previously studied diffusively coupled biophysical excitable systems, such as those in \cite{ambrosio2018global,ambrosio2019large,gambino2019pattern,jun2010spiral,kuznetsov2017pattern,liao2011pattern,madzvamuse2015stability,wang2008delay}. To the best of our knowledge, a clear analytical study describing the dynamics of a diffusively coupled, slow-fast, neuron model in which only the membrane voltage is spatially distributed, has not yet been deeply explored with respect to pattern formation and emergence of spirals. Often biophysical phenomena observed in experiments can be reproduced by solving numerically a nonlinear system of coupled PDEs that describes them. However, the underlying dynamical behavior of real biophysical systems generally involves a large number of variables.

The excitable FHR model studied here is an extended version of the FitzHugh-Nagumo (FHN) model \cite{fitzHugh1961}. Excitable, biophysical media represented by the slow-fast dynamics of electrically coupled FHR neurons \cite{izhikevich2001synchronisation,rinzel1987formal,rinzel1982bursting,wojcik2015voltage} give rise to wave propagation when the system becomes excited above a threshold, called a traveling wave, that travels through the nerve cells. The responses propagate along the axons, allowing the consideration of the system as a spatially distributed neural medium. The ideal biophysical nerve membrane model that produces bursting was formulated and studied numerically in \cite{rinzel1987formal,rinzel1982bursting}. The FHR model offers a wider view of the original FHN model by taking into consideration a special type of bursting, known as elliptic bursting. As a consequence, the FHR model with different neurocomputational properties exhibits various responses at different levels of biophysical plausibility, similar to real neurons.

Here, we consider diffusively coupled identical FHR neurons described by a system of three coupled PDEs. The model of the single deterministic nerve cell is studied theoretically and numerically. Different parameter regimes that correspond to qualitatively different biophysical behaviors are investigated. The self-sustained and self-organized behavior are investigated for a localized stimulus current with different synaptic couplings. The possible oscillatory behavior is considered and analytical conditions are derived. We find that weak stimulus injected to the system when it resides at the stable steady-state, produces nontrivial firing responses. Consequently, the stimulus perturbations can be applied to the system deterministically, leading to the emergence of different spatiotemporal phenomena.

Particularly, we explore the dynamics in a continuous, 1-dimensional piece of spatially extended neural cable \cite{bar1993turbulence,meier2015bursting,perc2007fluctuating} and in a square, considered as a 2-dimensional piece of spatially extended neural tissue. In the 1-dimensional case, neurons are coupled along a chain of excitable cells by means of a 1-dimensional diffusion term. We find suitable parameter regimes for which the system shows different fluctuations and pattern formations in the vicinity of two Hopf bifurcations for a certain range of coupling strengths. In the 2-dimensional case, we employ the amplitude equations and elucidate phenomena of spatiotemporal pattern formation. Our analysis is based on the complex Ginzburg-Landau equation (CGLE) \cite{aranson2002world,ipsen2000amplitude,ipsen2000amplitud,kuramoto1984springer,zemskov2011amplitude} for the derivation of the amplitude equations using multiple space and time scales. The method provides a Taylor series expansion of the original nonlinear equations with many power operators and depends on the expansion of the linear and nonlinear terms of a small perturbation parameter close to the onset of instability. The coefficients of CGLE obtained here, may have their own values which might relate to certain dynamical behavior in the diffusion process under different experimental conditions. We also derive analytically the stability of the coupled system.

We report on the presence of various emerging spatiotemporal patterns for different electrical impulses in the 2-dimensional system. The slow-fast dynamics supports the formation of spirals due to excitabilities of the slow time-scale in the single-neuron model. Various spiral patterns emerge in the diffusive system when it is not in a stable, stationary state. Spiral and target waves are known as ordered waves, which are often observed in the extended excitable system. Generally, suitable periodic forcing plays a major role in developing target waves \cite{jun2010spiral,ma2010transition,wu2013formation,hu2013selection}. Spiral or antispiral waves behave asymptotically as plane waves far from their cores \cite{nicola2004antispiral,gong2003antispiral} and a variety of complex patterns can manifest, such as target waves, spirals-antispirals, hexagons, spot-stripes, mixed patterns, etc., \cite{kang2020formation,kondo2010reaction,liu2020synchronisation,ma2013simulating,raghavachari1999waves,wang2017}. Spirals are particular types of patterns that rotate around a central point, known as the rotor. Rotors show powerful rhythmic activity by sending rotating-type, robust, wavy patterns outwards. The head of a stable rotating free spiral wave moves around a circular core. The motion of the spiral's tip controls the dynamics of spiral waves. Local excitatory connections play essential role in the dynamics and emergence of spirals during certain brain functions in the cortex \cite{huang2010spiral}. Antispiral waves observed near the Hopf bifurcation and the direction of wave-propagation are determined by the competition between waves and their surrounding bulk oscillations. The direction of the wave-propagation (either outwards or inwards) plays a role in characterizing spiral or antispiral waves. Spiral patterns may pinpoint to a mechanism that changes the irregular dynamics of cortical neurons to rhythmic behavior \cite{schiff2007dynamical}. Particularly, the mechanism changes the frequency of oscillations and amplitudes, as well as the spatial coherence activity. Our results suggest that spiral waves can emerge in a 2-dimensional spatially extended system with dynamics around a Hopf bifurcation. We report on the formation of single- and multi-arm spiral patterns. Spiral waves emerge when the target waves break-up and are associated with spatial heterogeneity in the system \cite{ma2013simulating}.

Spatiotemporal patterns can often be observed in coupled excitable oscillators, cardiovascular systems and in neural systems \cite{hu2013selection,jun2010spiral,kremmydas2002spiral,ma2010transition,witkowski1998spatiotemporal,wu2013formation}. It is important to understand the characteristics and formation of these spirals. It has been observed that stable spiral patterns can emerge in an oscillatory system around a Hopf bifurcation associated with the spatially uniform steady-state solution. The spiral waveforms have been investigated in bursting media in \cite{jiang2015formation}. These results may be related to the formation and development of spiral patterns in neural systems, especially in cortical areas in the brain. These types of spatiotemporal patterns may also be related to the characteristics of the generation of activity recorded in EEG signals, to epileptic seizures, to migraine-related issues and to the transmission of visual images in the cortex \cite{hu2013selection,huang2010spiral,jiang2015formation,keane2015propagating,liu2017transition,milton1993spiral}. Our results show that diffusively coupled FHR neurons with different firing characteristics, depending on control parameters, participate in the collective behavior of the system.

The paper is organized as follows: In Sec. \ref{Section2}, we discuss the single FHR model and perform a bifurcation analysis to reveal the multitude of dynamical behaviors as a function of the external current. In Sec. \ref{Section3}, we study, analytically and numerically, the effects of diffusion on coupled FHR neurons arranged in a 1-dimensional configuration. We also discuss the synchronization index that we use to quantify the collective behavior in the spatially extended systems. In Sec. \ref{Section4}, we study the diffusive properties of the dynamics of coupled FHR neurons, arranged in a 2-dimensional spatially extended mesh. We describe analytically the diffusive properties using the amplitude equation and present the results that reveal complex structures for positive diffusive couplings, such as target waves, spiral patterns, and two- and multi-arm spiral waves. Finally, in Sec. \ref{Section5}, we conclude our work and discuss it in the framework of other studies in the field.

\section{The excitable FHR model} \label{Section2}

The FHR model is the extended version of the original FHN model \cite{fitzHugh1961}. Incorporating slow dynamics, it can exhibit a wide range of different firing patterns in certain parameter regimes. The FHR model was introduced by FitzHugh and Rinzel \cite{izhikevich2001synchronisation,rinzel1987formal,rinzel1982bursting, wojcik2015voltage} and is also known as the elliptic bursting model \cite{wojcik2015voltage}. It is described by the system of ordinary differential equations
\begin{align}\label{model}
\dot{u} &= f(u)-v + w + I,\nonumber\\
\dot{v} &= \delta(a + u-bv),\\
\dot{w} &= \mu(c-u-w),\nonumber
\end{align}
where $u$ represents the membrane voltage of the nerve cell and, $v$ and $w$, the recovery and slow variables, respectively. Particularly, $v$ provides a slower negative feedback in the system. Here,
\begin{equation}
f(u)=u-\frac{u^3}{3}\label{f_u}
\end{equation}
is the cubic nonlinearity in the membrane voltage and allows for regenerative, self-excitation via positive feedback. $I$ indicates a constant external current stimulus and $\delta$, $a$, $b$, $\mu$ and $c$ are the parameters of the system. Parameter $\mu$ is a small time-scale number that controls the slow variable $w$ and $c$ plays a role similar to $a$ in the FHN model \cite{fitzHugh1961}. The decrease of $a$ and $c$ gives rise to longer intervals between consecutive bursting activities, with the system showing relatively fixed times in burst duration. Additionally, with the increase of $a$, the interburst intervals become shorter and the oscillations change to tonic spiking \cite{izhikevich2001synchronisation,wojcik2015voltage}. In our work, we have used $\delta=0.08$, $a=0.7$, $b=0.8$, $\mu=0.002$, $c=-0.775$ and we will be varying $I$.

To understand the dynamics of system \eqref{model}, we study its local and global stability properties performing a bifurcation analysis. To this end, its nontrivial fixed point is denoted by $E=(u_0,v_0,w_0)$, where $v_0=(a+u_0)/b$ and $w_0=c-u_0$. In this framework, $u_0$ can be obtained from the real solution to the cubic equation
\begin{equation}\label{zerosolution}
u_{0}^3+Pu_0+Q=0,
\end{equation}
where $P=3/b$ and $Q=-3(I-a/b+c)$. Its discriminant is given by $\Delta=-4P^3-27Q^2<0$, provided that $b>0$, which implies that Eq. \eqref{zerosolution} has only one real root. Particularly, the real solution to Eq. \eqref{zerosolution} is given by
\begin{equation}\label{eq_u_0}
u_0=-\frac{1}{3}\bigg(\Delta_2+\frac{\Delta_0}{\Delta_2}\bigg),
\end{equation}
where $\Delta_0=-3P$, $\Delta_2=\Bigg(\bigg(\Delta_1+\sqrt{\Delta_{1}^2-4\Delta_{0}^3}\bigg)/2\Bigg)^{\frac{1}{3}}$ and $\Delta_1=27Q$. This results to system \eqref{model} admitting the unique nontrivial fixed point $E=(u_0,v_0,w_0)$, where $u_0$, $v_0=(a+u_0)/b$ and $w_0=c-u_0$ are given by Eq. \eqref{eq_u_0}.

The Jacobian matrix, $J$, of system \eqref{model}, computed at the fixed point $E$, is given by
\begin{equation*}
J=\begin{pmatrix}
a_{11} & a_{12} & a_{13}\\
a_{21} & a_{22} & a_{23}\\
a_{31} & a_{32} & a_{33}
\end{pmatrix},
\end{equation*}
where $a_{11}=1-u_0^2$, $a_{12}=-1$, $a_{13}=1$, $a_{21}=\delta$, $a_{22}=-b\delta$, $a_{23}=0$, $a_{31}=-\mu$, $a_{32}=0$ and $a_{33}=-\mu$. The characteristic equation is given by 
\begin{equation}\label{chaequ}
\lambda^3 + a_1\lambda^2 + a_2\lambda + a_3=0,
\end{equation}
where $a_1=-1+b\delta+\mu+u_0^2$, $a_2=\delta-b\delta + b\mu\delta + b\delta u_0^2+ \mu u_0^2$ and $a_3=\mu\delta+b\mu\delta u_0^2$. Analyzing the stability of the system in a neighborhood of the fixed point $E$ using the Routh-Hurwitz criterion \cite{dejesus1987routh}, implies that $a_1>0$, $a_2>0$, $a_3>0$ and that the second order minor $a_1a_2-a_3>0$. Next, we prove a theorem that shows for which values of $a$, $b$ and $c$, the fixed point $E$ is globally asymptotically stable when $I=0$.

\begin{theorem}
The fixed point $E=(u_0,v_0,w_0)$ of system \eqref{model} is globally asymptotically stable in the exterior of the ellipsoid $$\frac{1}{3}\bigg(u^2-\frac{3}{2}\bigg)^2+b\bigg(v-\frac{a}{2b}\bigg)^2+\bigg(w-\frac{c}{2}\bigg)^2=\frac{3b+a^2+bc^2}{4b}$$ for applied current stimulus $I=0$.
\end{theorem}

\begin{proof}
We consider the Lyapunov function $L:\mathbb{R}^3\rightarrow\mathbb{R}$
\begin{equation*}
L(u,v,w)=u^2+\frac{1}{\delta}v^2+\frac{1}{\mu}w^2,
\end{equation*}
where $\delta$ and $\mu$ are positive parameters appearing in system \eqref{model}. Then,
\begin{itemize}
\item{$L(u,v,w)\geq0$ for all $(u,v,w)\in \mathbb{R}^3$,}
\item{$L(u,v,w)=0$ if and only if $(u,v,w)=(0,0,0)$, and}
\item{all sublevel sets of $L$ are bounded, i.e., $L(u,v,w)\rightarrow\infty\mbox{ as }(u,v,w)\rightarrow\infty$.}
\end{itemize}
This shows that $L(u,v,w)$ is positive definite. Then, the time-derivative of $L$ is given by
\begin{align*}
\frac{dL(u,v,w)}{dt}&=2\bigg(u\dot u+\frac{1}{\delta}v\dot v+\frac{1}{\mu}w\dot w\bigg)\nonumber\\
&=-2\Bigg(\frac{1}{3}\bigg(u^2-\frac{3}{2}\bigg)^2+b\bigg(v-\frac{a}{2b}\bigg)^2+\bigg(w-\frac{c}{2}\bigg)^2-\frac{3b+a^2+bc^2}{4b}\Bigg).
\end{align*}
It follows that $\frac{dL}{dt}$ is negative outside the ellipsoid
\begin{equation*}
\frac{1}{3}\bigg(u^2-\frac{3}{2}\bigg)^2+b\bigg(v-\frac{a}{2b}\bigg)^2+\bigg(w-\frac{c}{2}\bigg)^2=\frac{3b+a^2+bc^2}{4b}
\end{equation*}
and that 
\begin{equation*}
\frac{dL(0,0,0)}{dt}=0.
\end{equation*}
Thus, the fixed point $E$ is globally asymptotically stable for $I=0$ in the exterior of the specified ellipsoid.
\end{proof}
Following \cite{izhikevich2001synchronisation,wojcik2015voltage}, we consider the three sets of external currents $I$: (a) $I=0.2$ (set 1), (b) $I=0.43$ (set 2) and (c) $I=0.5$ (set 3) around the two Hopf bifurcation points of system \eqref{model} that we discuss next.

\subsection{Bifurcation analysis of the single FHR model}\label{subsec_bif_anal}

Here, we study the bifurcation properties of the single FHR model \eqref{model} considering the current, $I$, as the bifurcation parameter. We use MatCont, a Matlab package for the numerical continuation and bifurcation analysis of continuous, parameterized dynamical systems \cite{MatCont_paper}. We plot the results in Fig. \ref{bifu}.

Particularly, for $I<0.137$, system \eqref{model} has a stable focus node (solid cyan line), exhibiting a quiescent state (Fig. \ref{bifu}). The system changes its stability and a stable limit cycle appears around $I=0.137$, which corresponds to the supercritical Hopf bifurcation (HB1). As $I$ increases further, the steady state appears again through a second supercritical Hopf bifurcation (HB2) at $I=3.16298$. For $I>3.16298$, the system has a stable focus node (solid cyan line), which corresponds to the quiescent state shown in yellow in the fourth inset. For $0.137<I<3.16298$, the system exhibits different types of firing patterns, shown in blue, red and green in the insets in Fig. \ref{bifu}, such as elliptic-type bursting, mixed-mode oscillations and tonic spiking. With further increase in $I$, the system shows a quiescent state, shown as the yellow curve in Fig. \ref{bifu} for $I=3.8$. In the range $0.137<I<3.16298$, the limit cycle changes its stability through period-doubling bifurcations at $I=0.1925$, $I=0.4683$ and $I=3.1075$. The system has an unstable focus for $0.137<I<0.6$ and a saddle node for $0.7<I<2.6$ (dotted cyan line). Increasing $I$ even more, the system has an unstable focus node up to $I\approx3.16298$. The existence of period-doubling bifurcations in the system is one of the routes to chaos. The system is chaotic where both the fixed point and limit cycle are unstable. As an example, we calculated the Lyapunov exponents (LEs) for $I=0.4$, where both the fixed point and limit cycle are unstable and found that they converge approximately to the values $(0.000121, -0.004023, -0.522634)$, ordered in descending order. As the largest LE is positive, it confirms the system is chaotic. Further on, in Fig. \ref{bifu}, solid and cyan dotted lines denote the quiescent and oscillatory regions, respectively. The stable and unstable limit cycles are shown in solid magenta and dotted black lines, respectively.

\begin{figure*}[!ht]
\centering
\includegraphics[width=\textwidth,height=12cm]{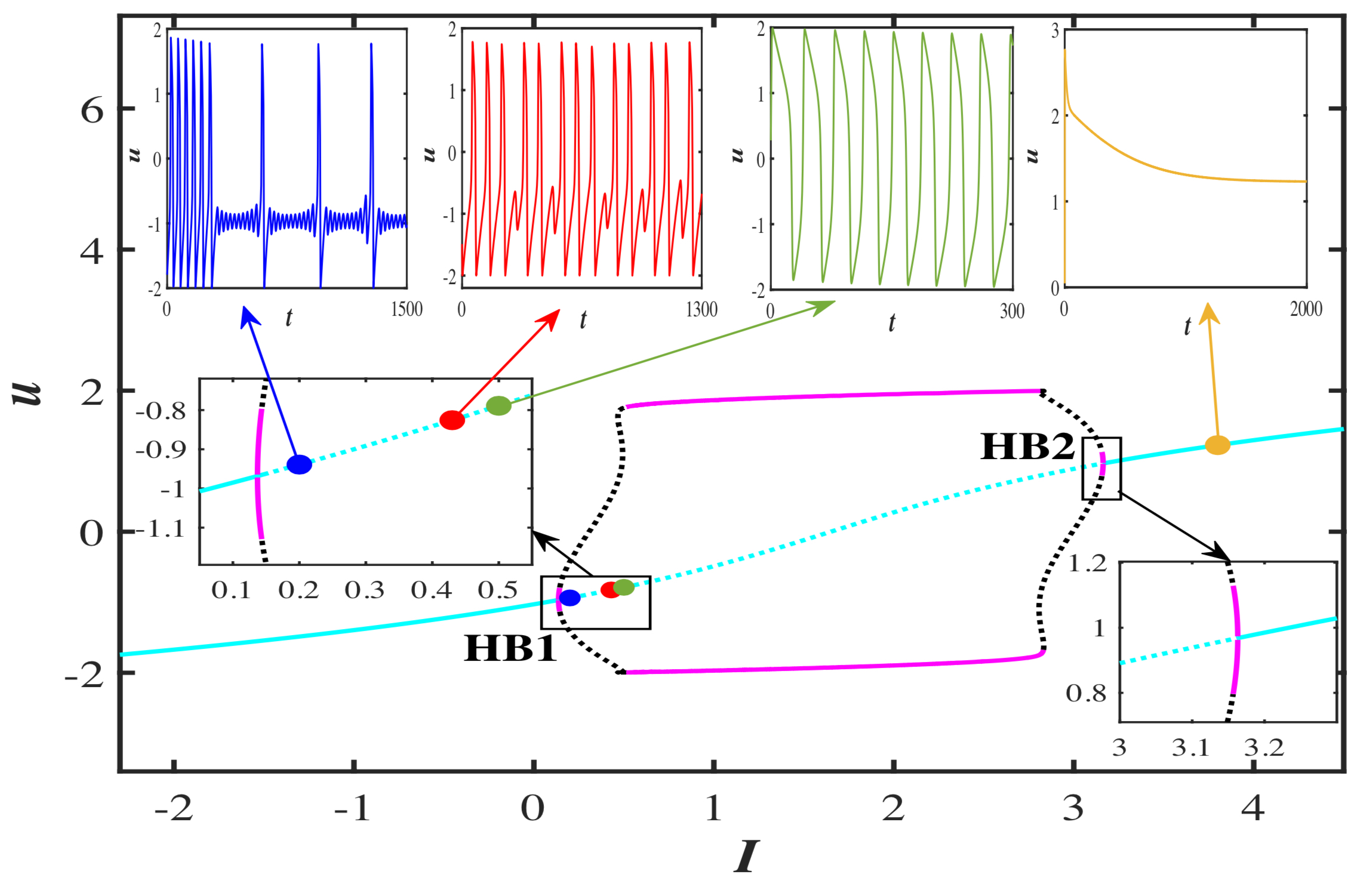}
\caption{Bifurcation diagram of the FHR system \eqref{model}, where solid and dotted cyan lines denote the quiescent and oscillatory regions, respectively. The stable and unstable limit cycles are shown in solid magenta and dotted black lines. Different oscillatory regimes and bifurcations as a function of the bifurcation parameter $I$ are also shown in the upper insets, where we present plots of the time-series $u$ as a function of time $t$ for the parameter sets 1, 2 and 3, shown in blue (elliptic-type bursting), red (mixed-mode oscillations), green (tonic spiking) and with further increase of $I$, it shows quiescent state. Note that HB1 and HB2 correspond to the two critical points $I=0.137$ and $I=3.16298$, respectively, where Hopf bifurcations occur in system \eqref{model}.}\label{bifu}
\end{figure*}

In the following, we verify our numerical results with respect to the critical bifurcation points HB1 and HB2 using an analytical approach \cite{liu1994}. The complex eigenvalues of the linearized FHR system of system \eqref{model} around the fixed point $E$ are given by $\gamma(I)=\alpha(I) \pm i\beta(I)$, where $i$ is the imaginary unit of the complex numbers. Suppose for a certain value of $I$, say $I\equiv I_0$, the following conditions are met: (i) $(I_0)=0$, (ii) $\beta(I_0)=\beta^*\neq0$ and (iii) $\frac{d\alpha(I)}{dI} \Big|_{I=I_0}=\alpha^* \ne 0$. Then, system \eqref{model} undergoes a Hopf bifurcation at $I\equiv I_0$. Conditions (i) and (ii) are known as the non-hyperbolicity conditions and condition (iii), as the transversality condition. To verify analytically the Hopf bifurcation at $I_0$, we consider $I$ as the bifurcation parameter, with $\alpha(I)$ and $\beta(I)$ given by
\begin{align*}
\alpha (I)&=0.31133-0.33333u_{0}^2+ (-68.6643 + 172.1666u_{0}-83.3333u_{0})/A^{\frac{1}{3}}-0.000333A^{\frac{1}{3}}
\end{align*}
and
\begin{equation*}
\beta (I)=\frac{118.9304-298.2023u_{0}^2 + 144.338u_{0}^4}{A^{\frac{1}{3}}}-0.000577A^{\frac{1}{3}},
\end{equation*}
where
\begin{align*}
A &= 10^5\bigg(931.04315-3529.5225u_{0}^2 + 3873.75u_{0}^4-1250u_{0}^6 \\&+ 10.392\sqrt {-67.14 + 25.33u_{0}^2 + 21.64u_{0}^4 + 606.6u_{0}^6-375.4u_{0}^8}\bigg),\\
u_{0}&= \frac{4.27494}{B}- 0.292402B,\\
B &= \bigg(99-60I + 1.414\sqrt{6463-5940I + 1800{I^2}}\bigg)^{\frac{1}{3}}.
\end{align*}

Thus, to find the critical points $I_0$ of system \eqref{model} where Hopf bifurcations occur, we find the values $I\equiv I_0$, where conditions (i)-(iii) hold. To this end, solving the equation $\alpha (I)=0$, we obtain $I=0.137$ and $I=3.16298$. This implies that
\begin{align*}
\alpha(0.137)&=\alpha(3.16298)=0,\\
\beta(0.137)&=\beta(3.16298)=-0.2792,\\
\frac{d\alpha(I)}{dI}\bigg|_{I=0.137}&=0.443,\\
\frac{d\alpha(I)}{dI}\bigg|_{I=3.16298}&=-0.443.
\end{align*}
Consequently, system \eqref{model} undergoes two Hopf bifurcations at the critical points $I=0.137$ (HB1) and $I=3.16298$ (HB2), which are in good agreement with the numerical results shown in Fig. \ref{bifu}.

In the next section, we study the diffusive properties of coupled FHR neurons arranged in a 1-dimensional configuration to investigate the collective dynamics and emergence of biologically plausible patterns around the two Hopf bifurcations.

\section{Diffusively coupled FHR neurons arranged in a 1-dimensional configuration}\label{Section3}

Here, we consider a 1-dimensional configuration of excitable, diffusively coupled FHR neurons through the membrane voltage $u$. The corresponding system is governed by the system of PDEs
\begin{align}\label{1-dimensionalmodel}
\frac{\partial u}{\partial t} &=f(u)-v + w + I+D\frac{\partial^2 u}{\partial x^2},\nonumber\\
\frac{\partial v}{\partial t} &=\delta(a + u-bv),\\
\frac{\partial w}{\partial t} &=\mu(c-u-w)\nonumber,
\end{align}
where $t$ is the time and $u(x,t=0)$, $v(x,t=0)$ and $w(x,t=0)$ for $x \in \Psi$, the initial conditions. $\Psi$ is the space of the 1-dimensional configuration and the function $f(u)$ is given by Eq. \eqref{f_u}. We consider that only the membrane voltage $u$ diffuses to neighboring neurons, modeled by the 1-dimensional diffusion term $D\frac{\partial^2 u}{\partial x^2}$, where $D\geq0$ is the diffusion coefficient or diffusion coupling or constant synaptic coupling strength. By 1-dimensional configuration we mean that the dimension of $\Psi$ is 1, thus $u$, $v$ and $w$ are functions of only one spatial coordinate, i.e., of $x$. Furthermore, we assume the zero-flux boundary condition,
\begin{equation}
\frac{\partial u}{\partial n}=\frac{\partial v}{\partial n} =\frac{\partial w}{\partial n}=0\label{zero-flux_bc}
\end{equation}
for $x \in \partial \Psi$, where $\partial \Psi$ is the boundary of $\Psi$. In this framework, $\frac{\partial}{\partial n}$ denotes the directional derivative along the outward normal $n$ to the boundary $\partial \Psi$. The reason for considering a zero-flux boundary condition is that it makes the 1-dimensional membrane impermeable at the two edges, i.e., no ions flow inside or outside the boundaries, and thus the membrane acts as an isolated cable \cite{meier2015bursting,mondal2019diffusion}. Next, we prove an important theorem.

\begin{theorem}
If the fixed point $E$ of the FHR system \eqref{model} is locally asymptotically stable, then the fixed point of the corresponding 1-dimensional diffusive model \eqref{1-dimensionalmodel} is also locally asymptotically stable.
\end{theorem}

\begin{proof}
We consider the particular solution of the linearized model
\begin{equation*}
\begin{pmatrix}
u\\
v\\
w
\end{pmatrix}=
\begin{pmatrix}
u_0\\
v_0\\
w_0
\end{pmatrix}+\epsilon
\begin{pmatrix}
u_k\\
v_k\\
w_k
\end{pmatrix}
e^{\lambda_k t + ikx} + c.c. + o(\varepsilon^2)
\end{equation*}
of system \eqref{1-dimensionalmodel}, where $c.c$. stands for ``complex conjugate'' terms, $\lambda_k$ is the wave length, $k>0$ the wave number along the $x$ direction and $x$ the directional vector in $\Psi$. The Jacobian matrix of system \eqref{1-dimensionalmodel} at the fixed point $E=(u_0,v_0,w_0)$ is given by
\begin{equation*}
J=\begin{pmatrix}
a_{11}-Dk^2 & a_{12} & a_{13}\\
a_{21} & a_{22} & a_{23}\\
a_{31} & a_{32} & a_{33}
\end{pmatrix}
\end{equation*}
and its characteristic equation by 
\begin{equation}\label{chaequ1-dimensional}
\lambda_k ^3 + b_1\lambda_k ^2 + b_2\lambda_k + b_3=0.
\end{equation}

The coefficients of Eq. \eqref{chaequ1-dimensional} are given by $b_1=a_1+Dk^2$, $b_2=a_2+b\delta Dk^2+ \mu Dk^2$ and $b_3=a_3+b\mu\delta Dk^2$, where $a_i$, $i=1,2,3$ are the coefficients of Eq. \eqref{chaequ}. We analyze the stability of the fixed point $E=(u_0,v_0,w_0)$ using the Routh-Hurwitz criterion \cite{dejesus1987routh}, resulting in the fixed point $E$ being stable if $b_1>0$, $b_2>0$, $b_3>0$ and if the second order minor $b_1b_2-b_3$ is positive. If we consider the stable fixed point of system \eqref{model}, then $b_1>0$, $b_2>0$, $b_3>0$ and $b_1b_2-b_3>0$ as $a_1>0$, $a_2>0$, $a_3>0$, $a_1a_2-a_3>0$, $D>0$ and $k^2>0$. Thus, the fixed point $E$ of the 1-dimensional diffusive model \eqref{1-dimensionalmodel} is locally asymptotically stable.
\end{proof}

It is worth it to note that if the fixed point $E$ of system \eqref{model} is unstable, then the fixed point of the corresponding 1-dimensional diffusive model \eqref{1-dimensionalmodel} can be made stable by increasing appropriately the value of the diffusion coefficient \cite{dubey2000}. For example, if we consider set 1, the fixed point $E=(-0.939127, -0.298909, 0.164127)$ of system \eqref{model} is unstable as $a_1=-0.0520405<0$, $a_2=0.0743373>0$, $a_3=0.000272891>0$ and $a_1a_2-a_3=-0.00414144<0$. Interestingly, for the diffusive model \eqref{1-dimensionalmodel}, the same fixed point is stable if $b_1=-0.0520405+Dk^2>0$, $b_2=0.0743373+0.066Dk^2>0$, $b_3=0.000272891+0.000128Dk^2>0$ and $b_1b_2-b_3=-0.00414144+0.0707747Dk^2 + 0.066D^2k^4>0$, i.e., when $Dk^2>0.05563$, meaning that $D>0.05563/k^2$, where $k^2>0$. Thus, the fixed point of the corresponding diffusive model \eqref{1-dimensionalmodel} can be made stable by increasing appropriately the diffusive coupling, $D$. One can prove analytically similar results for sets 2 and 3, where the corresponding conditions for the stability of the diffusive model \eqref{1-dimensionalmodel} are given by $Dk^2>0.255616$ and $Dk^2>0.315046$, respectively.

In Fig. \ref{1-dimensional_D_R_k}(a), the dashed blue, brown and green curves delineate the boundaries of the stable and unstable regions of sets 1, 2 and 3, and are given by
\begin{align}
D&=\frac{0.05563}{k^2},\label{Dk2_set1}\\
D&=\frac{0.255616}{k^2},\label{Dk2_set2}\\
D&=\frac{0.315046}{k^2},\label{Dk2_set3}
\end{align}
respectively. The solid blue, brown and green lines for sets 1, 2 and 3, respectively, were derived numerically and separate the stable from the unstable regions in  system \eqref{1-dimensionalmodel}. Particularly, the critical $D$ values where these lines are located, were computed using the PDEPE toolbox in Matlab and the bisection method. The method solves system \eqref{1-dimensionalmodel} numerically for different $D$ values until it converges to the critical $D$ at which the transition from unstable to quiescent (i.e., stable) dynamics occurs. For each $D$, we checked visually the dynamics based on plots such as those in Fig. \ref{1-dimensional_spatial}. The method resulted at  $D=6$ for set 1 (solid blue line), at $D=0.98$ for set 2 (solid brown line) and at $D=0.9$ for set 3 (solid green line).

%%%%%%%%%%%%%%%%%%%%%%%%%%%%%%%%%%%%%%%%%%%%%%%%%%%%%%
\begin{figure*}[!ht]
\centering
\includegraphics[width=\textwidth,height=13cm]{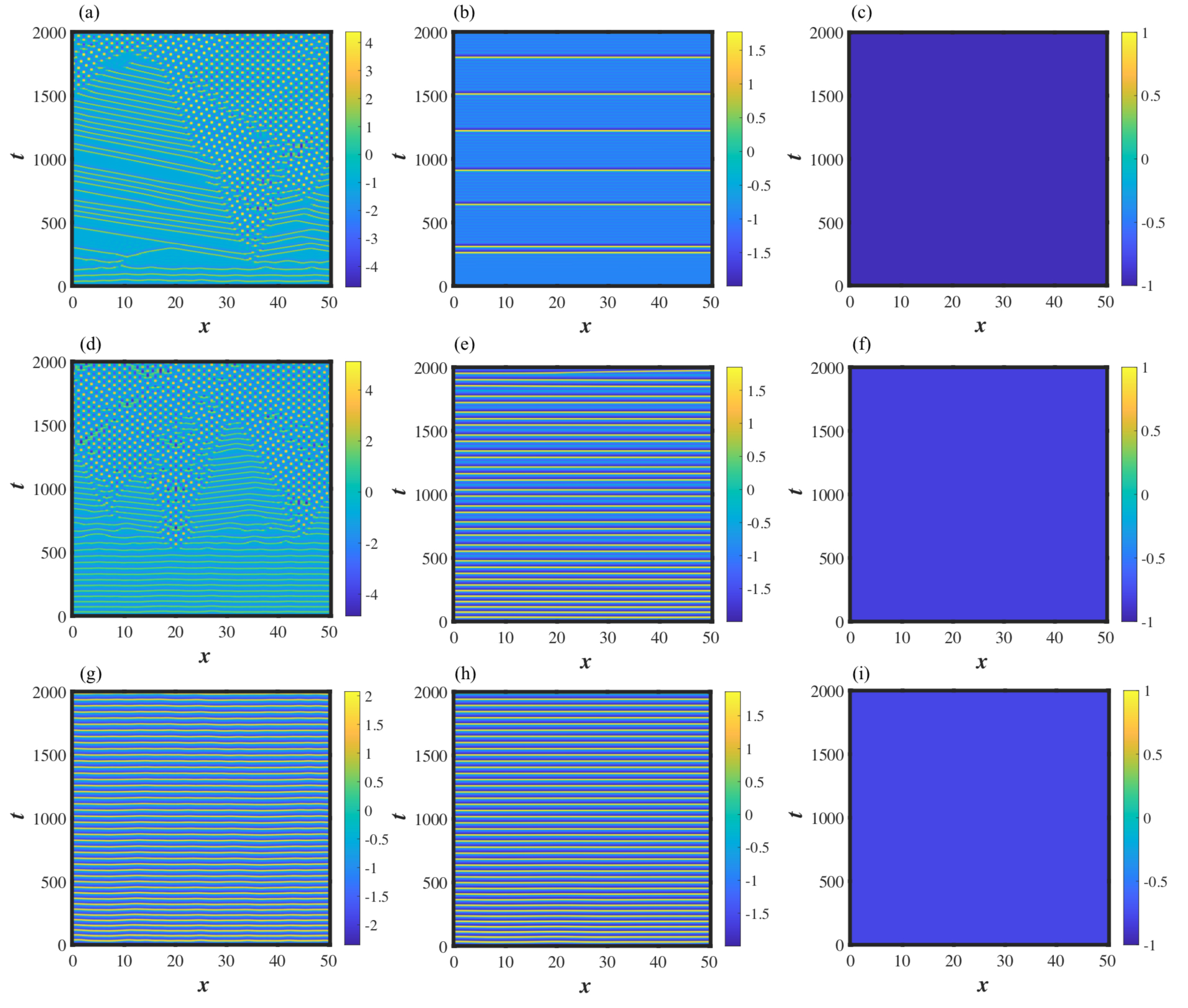}
\caption{Spatiotemporal plots of the system of coupled FHR neurons \eqref{1-dimensionalmodel} arranged in a 1-dimensional configuration for different values of the diffusion coefficient, $D$. Panels (a), (b) and (c) are for set 1 with $D=0.0001$, $D=1$ and $D=8$, respectively. Panels (d), (e) and (f) are for set 2 with $D=0.0001$, $D=0.1$ and $D=1$, respectively, and panels (g), (h) and (i) are for set 3 with $D=0.0001$, $D=0.01$ and $D=1$, respectively. The color bars encode the values of the membrane voltages $u(x,t)$ of diffusively coupled, FHR neurons in system \eqref{1-dimensionalmodel}.}\label{1-dimensional_spatial}
\end{figure*}
%%%%%%%%%%%%%%%%%%%%%%%%%%%%%%%%%%%%%%%%%%%%%%%%%%%%%%

%%%%%%%%%%%%%%%%%%%%%%%%%%%%%%%%%%%%%%%%%%%%%%%%%%%%%%
\begin{figure*}[!ht]
\centering
\includegraphics[width=\textwidth,height=13cm]{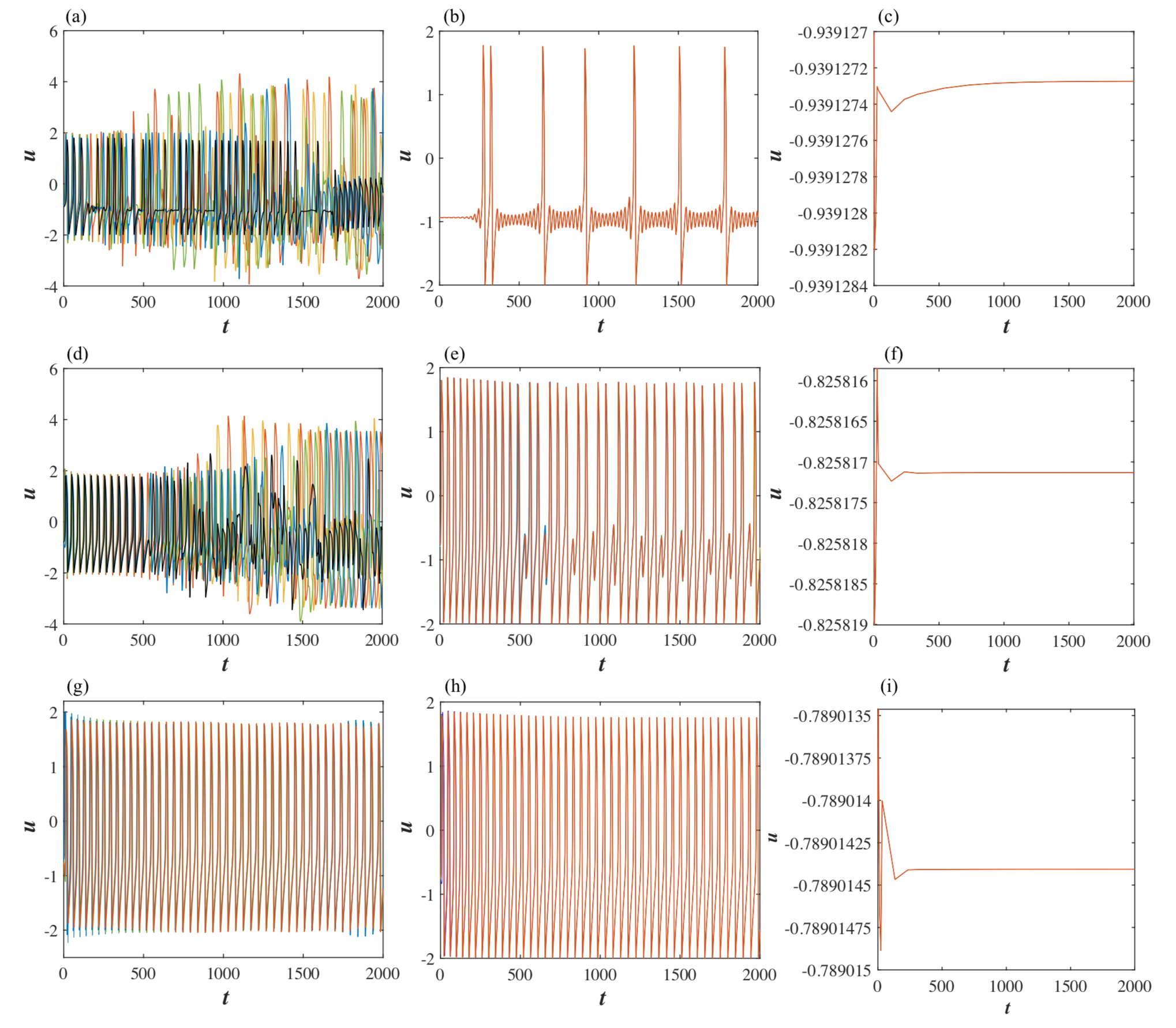}
\caption{Temporal evolution of arbitrarily chosen neurons of the system of coupled FHR neurons \eqref{1-dimensionalmodel} arranged in a 1-dimensional configuration for different values of the diffusion coefficient, $D$. The values of $D$ and sets of parameters are the same as those in the corresponding panels in Fig. \ref{1-dimensional_spatial}. The yellow, black, green, blue and red curves depict the temporal evolution of the arbitrarily chosen FHR oscillators in the coupled systems.}\label{1-dimensional_time}
\end{figure*}
%%%%%%%%%%%%%%%%%%%%%%%%%%%%%%%%%%%%%%%%%%%%%%%%%%%%%%

Next, we investigate numerically the effects of the diffusive coupling, $D$, on the formation of spatiotemporal patterns for the three sets of parameters. In the numerical simulations, we consider system \eqref{1-dimensionalmodel}, which is a system of coupled FHR neurons arranged in a 1-dimensional configuration, spaced at equal distances $\Delta x=0.1$, i.e., we consider $N_x=500$ spatial points and also a time-step $\Delta t=0.01$. Moreover, we consider the zero-flux boundary condition in Eq. \eqref{zero-flux_bc} and initial conditions in the vicinity of the fixed point $E$ of system \eqref{1-dimensionalmodel}. Hence depending on the value of $D$, $E$ can be either unstable or stable. We use the PDEPE toolbox in Matlab to solve numerically system \eqref{1-dimensionalmodel} and obtain its spatial behavior. The individual neurons for sets 1, 2 and 3 exhibit different voltage responses $u$ over time, i.e., firing patterns, which result in various collective dynamics that we study in Subsec. \ref{subsec_synchronisation} via the synchronization index $R$ of Eq. \eqref{eq_synch_index_R}.

First, for parameters in set 1, system \eqref{1-dimensionalmodel} for $D=0$ (which is essentially system \eqref{model}) exhibits mixed-mode type oscillations (elliptic bursting) as can be observed in blue in the left inset in Fig. \ref{bifu}. For smaller diffusion coupling, e.g., $D=0.0001$, the system exhibits inhomogeneous irregular firing, shown in Fig. \ref{1-dimensional_spatial}(a). However, as $D$ increases, it exhibits different firing patterns with bursting characteristics, known as elliptic bursting, a kind of mixed-mode type oscillations. It produces from typical bursting to rapid changes in membrane voltage oscillations, which involve simultaneous spike changes between low- and high-amplitude oscillations. This can be appreciated in Figs. \ref{1-dimensional_spatial}(b) and \ref{1-dimensional_time}(b). The yellow horizontal lines in Fig. \ref{1-dimensional_spatial}(b) indicate the peaks of the generated action potentials, shown in Fig. \ref{1-dimensional_time}(b). For higher diffusion coefficient $D=0.6$, system \eqref{1-dimensionalmodel} for set 1 becomes completely synchronized, confirmed by the synchronization index, $R=1$ that we discuss in Subsec. \ref{subsec_synchronisation}. At sufficiently higher diffusion coefficient $D=8$, all neurons show oscillation death as shown in Fig. \ref{1-dimensional_spatial}(c). In panels (a)-(c) in Fig. \ref{1-dimensional_time}, we show the corresponding time-series of arbitrary neurons for $D=0.0001$, particularly of neurons 1 and 8.

Next, for parameters in set 2, system \eqref{1-dimensionalmodel} for $D=0$ (i.e., system \eqref{model}) shows bursting activity as evidenced in the second inset from the left in Fig. \ref{bifu}. For weak coupling $D=0.0001$, it gives rise to inhomogeneity, which results in a train of irregular spikes, shown in Figs. \ref{1-dimensional_spatial}(d) and \ref{1-dimensional_time}(d). As the diffusion coefficient is increased to $D=0.1$, the system exhibits mixed-mode oscillations with different spike-numbers in single bursts, shown in Figs. \ref{1-dimensional_spatial}(e) and \ref{1-dimensional_time}(e). For even higher diffusion coefficient, e.g., for $D=1$, the system converges to a quiescent state and the dynamics shows a fixed point (Figs. \ref{1-dimensional_spatial}(f) and \ref{1-dimensional_time}(f)).

Last, we consider the case of tonic spiking for parameters in set 3. For $D=0.0001$, 0.01 and 1, spatiotemporal patterns and corresponding time-series of arbitrary neurons are shown in Figs. \ref{1-dimensional_spatial} (g)-(i) and \ref{1-dimensional_time} (g)-(i), respectively. In this case, system \eqref{1-dimensionalmodel} exhibits elliptic type bursting which comprises complex firing activities, more complex than bursting and spiking activities as it involves spikes and bunches of small-amplitude oscillations. This might be the reason set 1 needs higher diffusive coupling for the neurons to synchronize, compared to sets 2 and 3.

With the systematic change in $D$, the dynamics of system \eqref{1-dimensionalmodel} transits from the regime of inhomogeneous instability to a uniform steady-state through the formation of regular structures of trains of tonic spiking at intermediate $D$ values. For sets 2 and 3 and intermediate diffusion coefficients, $D=0.1$ and $D=0.01$, all neurons become synchronized. This synchronization property is verified by the synchronization index $R$ that we discuss in Subsec. \ref{subsec_synchronisation}.

The fixed point $E$ of system \eqref{model} is unstable for the values of the diffusion coefficient $D$ in panels (a), (b), (d), (e), (g), (h) in Figs. \ref{1-dimensional_spatial} and \ref{1-dimensional_time}. The stabilizing mechanism responsible for the stable dynamics in panels (c), (f), (i) in Figs. \ref{1-dimensional_spatial} and \ref{1-dimensional_time} results from the stability properties of the fixed point $E$ of systems \eqref{model} and \eqref{1-dimensionalmodel}. As we discussed before, the fixed point $E$ of system \eqref{1-dimensionalmodel} can be made stable by increasing appropriately the value of the diffusion coefficient. The diffusion coefficient $D$ in panels (c), (f), (i) in Figs. \ref{1-dimensional_spatial} and \ref{1-dimensional_time} is equal to 8, 1, 1, respectively, and for these values, the fixed point $E$ of system \eqref{1-dimensionalmodel} is stable. Hence starting with initial conditions in the vicinity of the stable fixed point results in the solutions remaining in its vicinity. This is shown in Fig. \ref{1-dimensional_time}(c), (f), (i) for the three sets of parameters. However, that is not the case for smaller $D$ values where the fixed point $E$ of system \eqref{1-dimensionalmodel} is unstable as can be seen in panels (a), (b), (d), (e), (g), (h) in Figs. \ref{1-dimensional_spatial} and \ref{1-dimensional_time}. In these cases, solutions that start in the vicinity of the unstable fixed point $E$ drift away in time (see panels (a), (b), (d), (e), (g), (h) in Fig. \ref{1-dimensional_time}). Hence to stabilize system \eqref{1-dimensionalmodel} in the vicinity of the fixed point $E$, one must find appropriately big values of the diffusion coefficient, $D$, for which solutions that start in the vicinity of $E$, remain in its vicinity for $t>0$. Consequently, the diffusion coefficient, $D$, plays a major role in the stabilization mechanism in system \eqref{1-dimensionalmodel} as it determines the stability properties of the fixed point $E$.

\subsection{Synchronization index}\label{subsec_synchronisation}

To quantify the collective behavior in systems of coupled oscillators (in our case of diffusively coupled FHR neurons), a synchronization measure based on mean field theory, can be defined. The synchronization index, $R$, ranges in $[0,1]$ \cite{liu2020synchronisation,wang2017}, where 0 indicates complete desynchronization (i.e., all oscillators are completely desynchronized) and 1 complete synchronization (i.e., all oscillators are completely synchronized). The synchronization index, $R$, is defined by
\begin{equation}\label{eq_synch_index_R}
R=\frac{\langle F^2 \rangle-\langle F \rangle^2}{\frac{1}{N_x} \sum_{i=1}^{N_x}(\langle u_i^2 \rangle-\langle u_i \rangle^2)},
\end{equation}
where
\begin{equation*}
F=\frac{1}{N_x} \sum_{i=1}^{N_x} u_i
\end{equation*}
is the average of the membrane voltages $u_i$ over all spatial points and $N_x$ the number of spatial points. The notation $\langle\cdot\rangle$ denotes the mean value of the argument over time.

The results for the synchronization index, $R$, as a function of the diffusion coefficient, $D$, are shown in Fig. \ref{1-dimensional_D_R_k}(b) for sets 1, 2 and 3, in blue, brown and green, respectively. For sets 1 and 2, the coupled neurons are desynchronized for weak diffusive coupling, $D$, as $R$ is close to 0. This indicates that all neurons exhibit almost desynchronized oscillations as shown in the spatial plots in panels (a) and (d) in Fig. \ref{1-dimensional_time}. However, as $D$ increases, all neurons of the three sets exhibit complete synchronization, evidenced in panels (b), (c), (e) and (f) in Fig. \ref{1-dimensional_time}. Interestingly, $R$ is almost equal to 1 for weak coupling $D$ for set 3, which can be appreciated in Fig. \ref{1-dimensional_D_R_k}(b).

\begin{figure*}[!ht]
\centering
\includegraphics[width=\textwidth,height=6cm]{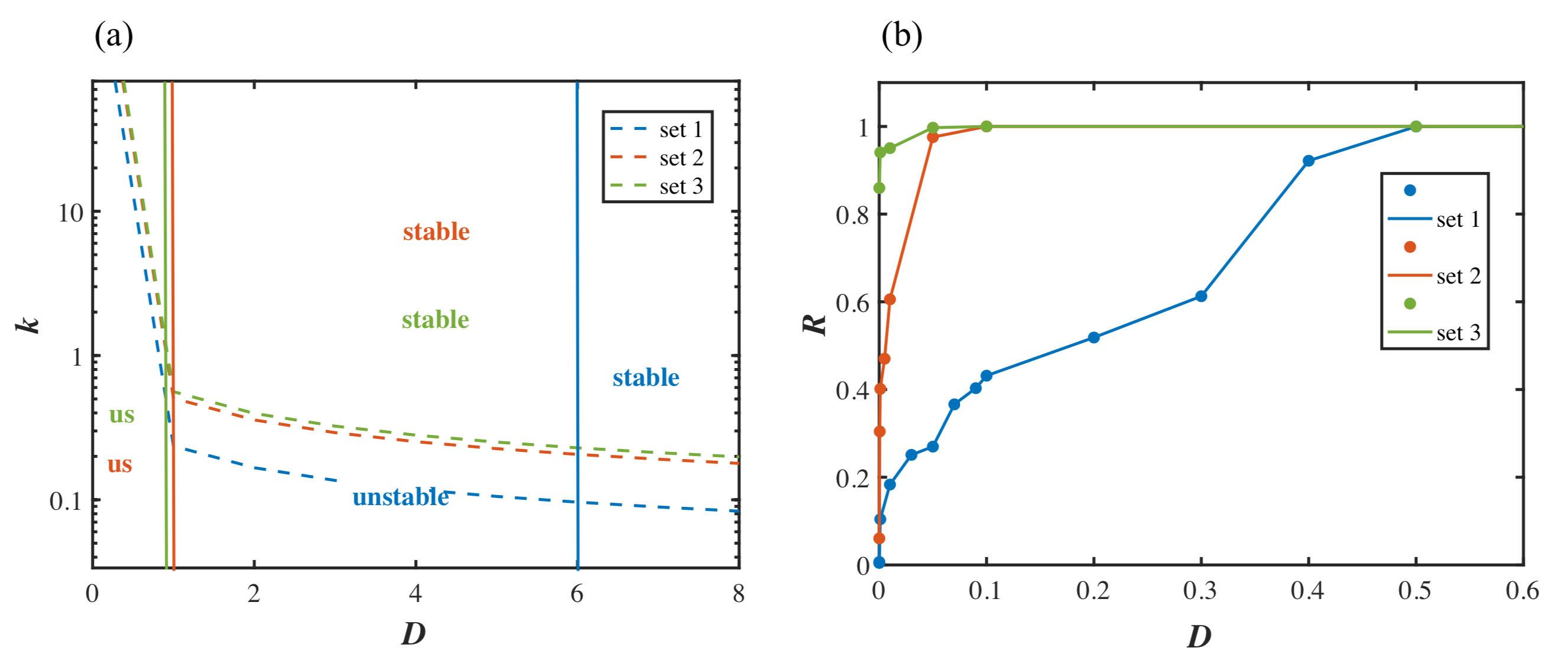}
\caption{Boundaries of stable and unstable regions and synchronization index, $R$, for the sets of parameters 1, 2 and 3. (a) The dashed blue, brown and green curves (see Eqs. \eqref{Dk2_set1}, \eqref{Dk2_set2}, \eqref{Dk2_set3}) indicate the boundaries of the stable and unstable regions of model \eqref{1-dimensionalmodel} for the parameter sets 1, 2 and 3, respectively. The solid blue, brown and green lines are located at $D=6$ (set 1), $D=0.98$ (set 2) and $D=0.9$ (set 3), are derived numerically with the help of the bisection method and, separate the stable from the unstable regions. Here, ``us'' stands for unstable region. (b) The synchronization index, $R$, as a function of the diffusion coefficient, $D$, for the sets of parameters 1, 2 and 3, shown in blue, brown and green, respectively. Note that in panel (b), the line segments connect the numerically computed values of $R$ plotted by small filled circles.}\label{1-dimensional_D_R_k}
\end{figure*}

\section{Diffusively coupled FHR neurons arranged in a 2-dimensional configuration}\label{Section4}

In this section, we investigate different types of pattern formations for diffusively coupled FHR neurons in a square domain, $\Psi$. In this case, the mathematical model is given by the system of coupled PDEs
\begin{align}
\frac{{\partial u}}{{\partial t}}&=f(u)-v + w + I+D\Bigg(\frac{{\partial^2 u}}{{\partial x^2}}+\frac{{\partial^2 u}}{{\partial y^2}}\Bigg),\nonumber\\
\frac{{\partial v}}{{\partial t}}&=\delta \left({a + u-bv} \right),\label{2-dimensionalmodel}\\
\frac{{\partial w}}{{\partial t}}&=\mu \left({c-u-w} \right),\nonumber
\end{align}
where we use similar initial and boundary conditions as in Sec. \ref{Section3}, adapted to the case. Thus, here $u\equiv u(x,y,t)$, $v\equiv v(x,y,t)$ and $w\equiv w(x,y,t)$, where $x$ and $y$ are the spatial coordinates and $t$ is the time. System \eqref{2-dimensionalmodel} can be written in the compact form
\begin{equation}\label{2-dimensionalmodel_compact_form}
\frac{{\partial X}}{{\partial t}}=H(X,I)+D_1\nabla^2X,
\end{equation}
where
\begin{equation*}
X=
\begin{pmatrix}
X_{1}\\
X_{2}\\
X_{3}
\end{pmatrix}
=
\begin{pmatrix}
u\\
v\\
w
\end{pmatrix},\;
H=
\begin{pmatrix}
u-\frac{{u^3}}{3}-v + w + I\\
\delta({a + u-bv})\\
\mu({ c-u-w})
\end{pmatrix}
\end{equation*}
and
\begin{equation*}
D_1=
\begin{pmatrix}
D & 0 & 0\\
0 & 0 & 0\\
0 & 0 & 0
\end{pmatrix}.
\end{equation*}
Equation \eqref{2-dimensionalmodel_compact_form} will prove useful in the next, where we discuss the amplitude equations that we use to study analytically the formation of spatiotemporal patterns in the vicinity of the two Hopf bifurcation points, HB1 and HB2.

\subsection{Amplitude equations}\label{AE}

The dynamics of system \eqref{2-dimensionalmodel} in the vicinity of a Hopf bifurcation can be described by the amplitude equations, also known as CGLE. It's general form is given by \cite{ipsen2000amplitude,ipsen2000amplitud,zemskov2011amplitude}
\begin{equation}\label{CGLE}
\frac{\partial W}{\partial t}=W-(1+i\alpha){\big|W\big|}^2W+(1+i\beta)\nabla^2W,
\end{equation}
where $i$ denotes the imaginary unit of the complex numbers, i.e., $i^2=-1$ and $\nabla^2$ is the 2-dimensional Laplacian operator. As the real parameters $\alpha$ and $\beta$ vary, the complex amplitude, $W$, exhibits rich dynamics. In the following, we seek to find the values of $\alpha$ and $\beta$ in the vicinity of the two Hopf bifurcations, HB1 and HB2 (as discussed in Subsec. \ref{subsec_bif_anal}), i.e., around the bifurcation points $I=0.137$ and $I=3.16298$.

The Jacobian matrix $J$ in this case has two complex conjugate, imaginary eigenvalues $\lambda_{1,2}=\pm\omega i$. The right eigenvectors that correspond to the eigenvalues $\pm i\omega$ are denoted by $U_1$, $U_2=\bar{U}_1$ and the left by $U_1^+$, $U_2^+=\bar{U}_1^+$. For simplicity, we consider $U=U_1=\bar{U}_2$ and $U^+=U_1^+=\bar{U}_2^+$. Furthermore, the right and left eigenvectors can be normalized according to $U_i^+U_j=\delta_{ij}$, where $\delta_{ij}$ is the Kronecker delta, i.e., $U^+U=\bar{U}^+\bar{U}=1$ and $U^+\bar{U}=\bar{U}^+U=0$. In Subsec. \ref{subsec_bif_anal}, we have shown that the two Hopf bifurcations, HB1 and HB2, occur at $I=0.137$ and $I=3.16298$, respectively. Interestingly, the eigenvalues of $J$ for both are the same and are given by $\lambda_{1,2}=\pm 0.279302i$ and $\lambda_3=-0.0036$ as $J$ contains the term $u_0^2$ and its value is the same for both bifurcation points (i.e., $u_0=\pm 0.96829$). The right and left eigenvectors that correspond to the eigenvalue $\lambda_1=0.279302i$ are $U=(0.963143, 0.0600545-0.262111i, -0.00005+0.0069i)^T$ and $U^+=(0.5198-0.1159i, -0.0108 + 1.8586i, -0.4018-1.8639i)$, respectively. These again refer to both bifurcation points as the Jacobian matrix $J$ and its eigenvalues are the same for both.

To compute $\alpha$ and $\beta$ in Eq. \eqref{CGLE}, we start with the unscaled Ginzburg-Landau equation \cite{ipsen2000amplitude,ipsen2000amplitud}
\begin{equation}
\dot{W}=\sigma_1 IW-g\big|W\big|^2W + d\nabla^2W \label{uscgle}
\end{equation}	
and seek to find the values of $\sigma_1$, $g$ and $d$ using the values in Table 1 in \cite{ipsen2000amplitude,ipsen2000amplitud}. Particularly, parameters $\sigma_1$, $g$ and $d$ are given by
\begin{align*}
\sigma_1&=U^+H_{XI}U+U^+H_{XX}(U,l_{001}),\\ 
g&=-\Big(U^+H_{XX}(U,l_{110}) + U^+H_{XX}(\bar{U},l_{200})+\frac{1}{2}U^+H_{XXX}(U,U,\bar{U})\Big),\\
d&=U^+D_1U,
\end{align*}
respectively, where
\begin{align*}
H_{{XX}} \left({\xi ,\,\eta} \right)&=\sum\limits_{i,j=1}^3 {\left. {\frac{{\partial^2 H}}{{\partial X_i \partial X_j}}} \right|_{E=(u_0,v_0,w_0)} \xi_i \eta_j},\\
H_{{XXX}} \left({\xi ,\,\eta ,\,\zeta} \right)&=\sum\limits_{i,j,k=1}^3 {\left. {\frac{{\partial^3 H}}{{\partial {X}_i \partial {X}_j \partial {X}_k}}} \right|}_{E=(u_0,v_0,w_0)} \xi_i \eta_j \zeta_k,\\
l_{001}&=(0.4571, 0.5714, -0.4571)^T,\\
l_{110}&=(0.8212, 1.0265, -0.8212)^T,\\
l_{200}&=(-0.2395-2.1447i,-0.3070-0.0009i, 0.0077-0.0008i)^T.
\end{align*}
In our case,
\begin{equation*}
H=
\begin{pmatrix}
u-\frac{u^3}{3}-v + w + I\\
0.08(0.7 + u-0.8v)\\
0.002(-0.775-u-w)
\end{pmatrix},
\end{equation*}
thus
\begin{equation*}
\frac{\partial^2 H}{\partial X_1^2}\bigg|_{E}=
\begin{pmatrix}
-2u_0\\
0\\
0
\end{pmatrix}
\mbox{ and }\;
\frac{\partial^3 H}{\partial X_1^3}\bigg|_{E}=
\begin{pmatrix}
-2\\
0\\
0
\end{pmatrix},
\end{equation*}
with all other terms in $H_{XX}$ and $H_{XXX}$ being equal to zero. With this in mind, we obtain $\sigma_1=0.4432-0.0988i$, $g=0.3642 + 2.1015i$ and $d=D(0.5006-0.1117i)$. Next, using the transformation
\begin{align}
W&=\sqrt{I}W^{\prime}\nonumber,\\
t&=t^{\prime}/I\nonumber,\\
x&=x^{\prime}/\sqrt{I}\label{CGLE_transformation},\\
y&=y^{\prime}/\sqrt{I}\nonumber,
\end{align}
we obtain
\begin{equation}
\frac{\partial W}{\partial t}=\sigma_1 W-g|W|^2W + d\nabla^2W, \label{scgle}
\end{equation}
which is known as the scaled complex Ginzburg-Landau equation. Transformation \eqref{CGLE_transformation} renders CGLE independent of the distance from the bifurcation point. The dynamics of system \eqref{2-dimensionalmodel} can be described by the solution of Eq. \eqref{scgle}, except for a short initial decay of transient eigenmodes. Finally, introducing the new transformation
\begin{align*}
W&=\sqrt{\frac{\Re(\sigma_1)}{\Re(g)}}W^{\prime}e^{i\frac{\Im(\sigma_1)}{\Re(\sigma_1)}t^{\prime}},\\
t&=\frac{t^{\prime}}{\Re(\sigma_1)},\\
x&=\sqrt{\frac{\Re(d)}{\Re(\sigma_1)}}x^{\prime},\\
y&=\sqrt{\frac{\Re(d)}{\Re(\sigma_1)}}y^{\prime},
\end{align*}
we obtain the dimensionless CGLE \eqref{CGLE}, that is the equation
\begin{equation}
\frac{\partial W}{\partial t}=W-(1+i\alpha){\left| {W} \right|}^2W+(1+i\beta)\nabla^2W,\label{eq10}
\end{equation}
where $\alpha=\frac{\Im(g)}{\Re(g)}=5.7706$ and $\beta=\frac{\Im(d)}{\Re(d)}=-0.2230$. In this framework, $\Re(*)$ and $\Im(*)$ denote the real and imaginary parts of the argument, respectively. Clearly, $\alpha + \beta=5.5476 >0$, which confirms the existence of antispiral patterns in the vicinity of the two Hopf bifurcations \cite{nicola2004antispiral}, HB1 and HB2, at the critical points $I=0.137$ and $I=3.16298$, respectively.

\subsection{Results}

Here, we focus on the diffusively coupled FHR model \eqref{2-dimensionalmodel} and study numerically the emergence of complex structures in a square spatial domain $\Psi$ for $D>0$, considering $I$ as the varying parameter. The method of amplitude equations was employed in Subsec. \ref{AE} to show the existence of antispiral patterns in the vicinity of the two Hopf bifurcations, HB1 and HB2. In the following, we discuss specific patterns that emerge around HB1.

To do so, we construct a spatially extended medium with 2-dimensional spatial diffusion using an $N_x \times N_y=100 \times 100$ mesh grid. We solve numerically the resulting system using the forward Euler method with central differences, a finite-difference scheme, with spatial step size $\Delta=\Delta x=\Delta y=1.25$ and integration time-step $\Delta t=0.1$. The initial conditions are considered with appropriate periodic perturbations from the initial conditions of the system in Sec. \ref{Section2}, using similar boundary conditions as in the 1-dimensional case, adapted to the case. We denote by $u_{i,j}(t)$ the membrane voltage of a neuron at the node $(i, j)$ on the grid at time $t$. Consequently, the sum of the spatial, second-order partial derivatives $\frac{{{\partial^2}u}}{{\partial {x^2}}} + \frac{{{\partial^2}u}}{{\partial {y^2}}}$ can be approximated by
\begin{equation*}
\frac{{{\partial^2}u}}{{\partial {x^2}}} + \frac{{{\partial^2}u}}{{\partial {y^2}}}=\frac{1}{{{\Delta^2}}}\left({{u_{i-1,\,j}} + {u_{i + 1,\,j}} + {u_{i,\,j-1}} + {u_{i,\,j + 1}}-4{u_{i,\,j}}} \right).
\end{equation*}
We note that in the following, we will study the spatiotemporal characteristics in the context of the nonlinear, coupled, diffusive system \eqref{2-dimensionalmodel}, where the coupling indicates the synaptic, diffusive coupling among neurons on the nodes of the $N_x \times N_y$ grid. Furthermore, we construct a square $N_x \times N_y$ domain $\Psi$ (grid) using nearest neighbor connectivities and explore the evolution of target waves and spirals for fixed values of the diffusion coefficient, $D$, using long-time numerical integration, i.e., up to $t=20000$.

\begin{figure*}[!ht]
\centering\includegraphics[width=\textwidth,height=13cm]{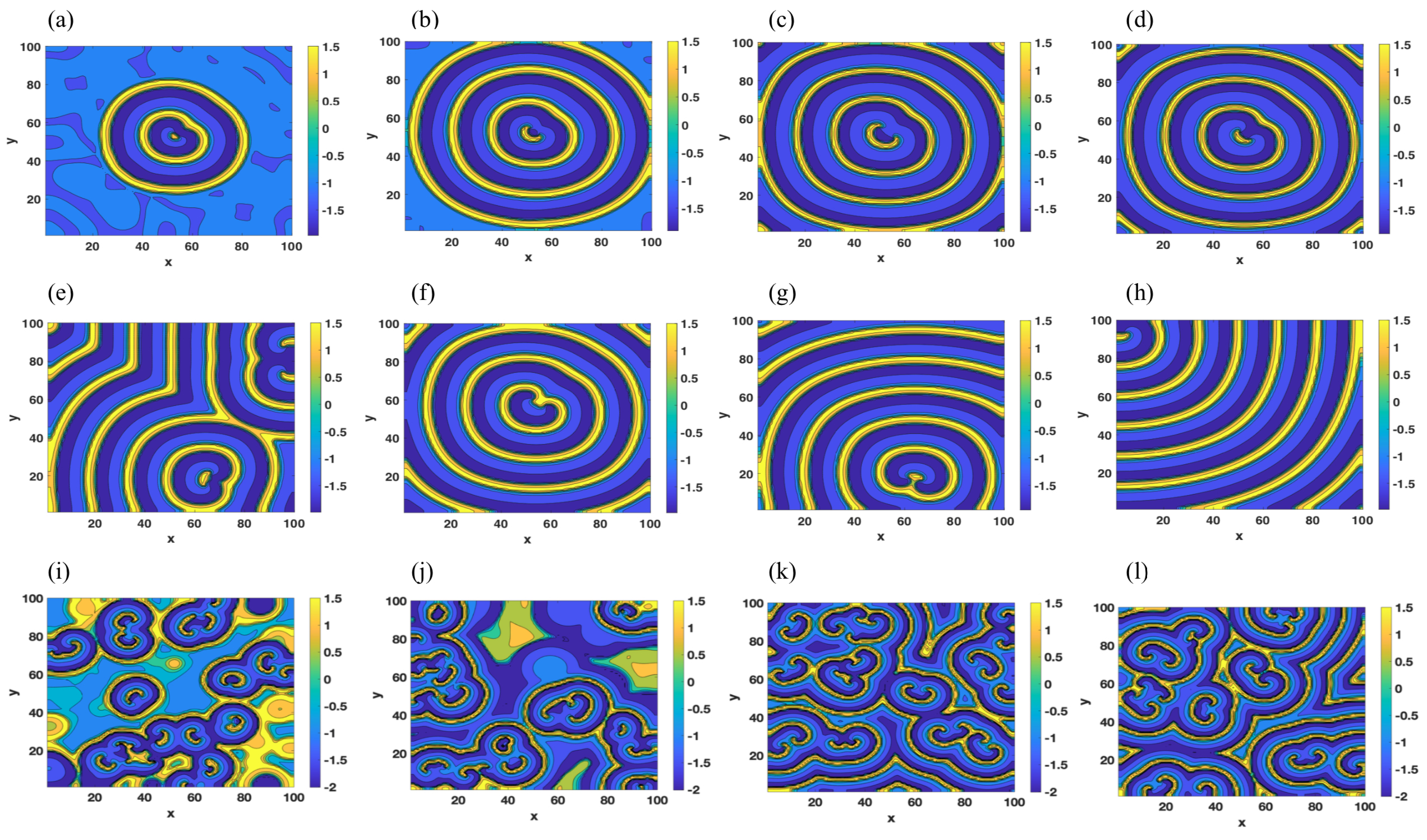}
\caption{Formation of target waves and spiral patterns resulting from the system of diffusively coupled FHR neurons \eqref{2-dimensionalmodel} arranged in a 2-dimensional space for set 1 with $D=0.25$ at $t=3450$, 3500, 20000, 40000 in panels (a) to (d), set 2 with $D=0.3$ at $t=2000$, 3000, 5000, 40000 in panels (e) to (h) and set 3 with $D=0.009$ at $t=3500$, 5000, 20000, 40000 in panels (i) to (l). The color bars encode the values of the membrane voltages $u(x,y,t)$ of the diffusively coupled FHR neurons of system \eqref{2-dimensionalmodel}.}\label{pattern}
\end{figure*}

We start with set 1, for which the single-neuron model \eqref{model} with no diffusion, exhibits mixed-mode oscillations, often known as elliptic type bursting \cite{izhikevich2001synchronisation}. For $D=0.25$, system \eqref{2-dimensionalmodel} gives rise to the emergence of target waves at $t=3450$ shown in Fig. \ref{pattern}(a). We also find that these target waves are sparse as shown in the same panel. The reason might be that most of the neurons cannot activate neighboring neurons in short time intervals. However, as time increases to $20000$, small spiral-like activity can be seen, bounded by target waves \cite{hu2013selection}. The developed target waves become denser, occupying the whole extend of the mesh grid, as shown in panels (b) and (c) in Fig. \ref{pattern}. The coexistence of spirals and target waves can be observed at $t=40000$ in Fig. \ref{pattern}(d), which suggests the patterns in the spatial domain are stable.

Next, we consider set 2 with $D=0.3$, where individual neurons exhibit mixed-mode oscillations. At the initial stage, multi-arm spiral waves emerge from broken waves, which can be observed in Fig. \ref{pattern}(e). As the time increases to $t=3000$, target waves and various wavefronts emerge as can be observed in Fig. \ref{pattern}(f). The spiral-like patterns coexist with target waves until $t=5000$, shown in Fig. \ref{pattern}(g), and the spirals are not stable. The reason may be the excitability of neurons in the arm segments of small spirals being high, attracting other small segments of spiral arms and forming more stable wavefronts. Interestingly, at even higher times, for example at $t=40000$, only spirals can be observed, shown in Fig. \ref{pattern}(h), with the centers of the wavefronts frequently changing over time. A group of small-size spiral profiles emerge that change into multi-arm spiral waves. Our numerical results show the emergence of different wavefronts that develop and break over time, with spirals emerging thereafter.

Multi-arm spiral waves emerge for set 3 for small diffusion coefficient values. For example, at $D=0.009$, target waves appear to cover partially the grid. Initially, we observe the formation of distinct clusters consisting of small curls on the blue backdrop of the mesh grid, that can be seen in panels (i) and (j) in Fig. \ref{pattern}. However, as time increases, for example at $t=20000$ and $t=40000$, multi-arm spirals emerge, shown in panels (k) and (l) in Fig. \ref{pattern}. This shows that the small curls collide with neighboring curls that extend to the whole mesh. With further increase in time,  multi-arm spirals emerge.

Next, we discuss the effect of diffusive coupling on the formation of target waves and spirals for long integration times. At coupling strength $D=0.05$, the formation of two-arm spiral waves \cite{hu2013selection} can be observed at $t=20000$, shown in Fig. \ref{patternD}(a) for set 1. Only a few spirals are visible. However, as $D$ increases slightly, for $D=0.09$ and $D=0.1$, antispiral waves emerge, shown in panels (b) and (c) in Fig. \ref{patternD}. The results are verified using Eq. (\ref{eq10}). It follows from the study of the amplitude equations that the antispirals exist in the vicinity of the two supercritical Hopf points and that they depend on $\alpha$ and $\beta$. Solving the system numerically, we can see the antispirals for suitable values of the diffusive coupling, $D$, depending on $\alpha$ and $\beta$. These multi-arm antispiral waves with an increasing number of arms extend to the whole spatial grid. For even higher diffusion values, for example for $D=0.245$, double spirals can be observed, shown in Fig. \ref{patternD}(d). Periodic forcing and coupling strength play a major role in promoting target waves by suppressing spiral waves \cite{wu2013formation,hu2013selection}. Increasing slightly the diffusive coupling to $D=0.25$, we observe the same route of transition, i.e., two spirals that collide and generate target wave patterns, shown in the first panel in Fig. \ref{pattern}.

When considering set 2 and for smaller couplings, for example for $D=0.1$, multi-arm spirals emerge. For even higher $D$, for example for $D=0.25$ and $D=0.3$, the multi-arm spiral waves dissolve and, target waves and spirals emerge, shown in panels (e) to (g) in Fig. \ref{patternD}. With the further increase of $D$ to 0.9, only spiral patterns can be observed, shown in Fig. \ref{patternD}(h). 

For set 3, with increasing $D$, multi-arm spirals emerge, occupying a greater extend on the spatial grid, observed in panels (i) to (l) in Fig. \ref{patternD}. We cross-checked the appearance of similar spatiotemporal patterns on a bigger, $N_x\times N_y=500\times500$ mesh grid and we obtained similar results, with no significant changes. Finally, the same wavy patterns emerged when running simulations for even longer times, i.e., up to $t=50000$.

\begin{figure*}[!ht]
\centering\includegraphics[width=\textwidth,height=13cm]{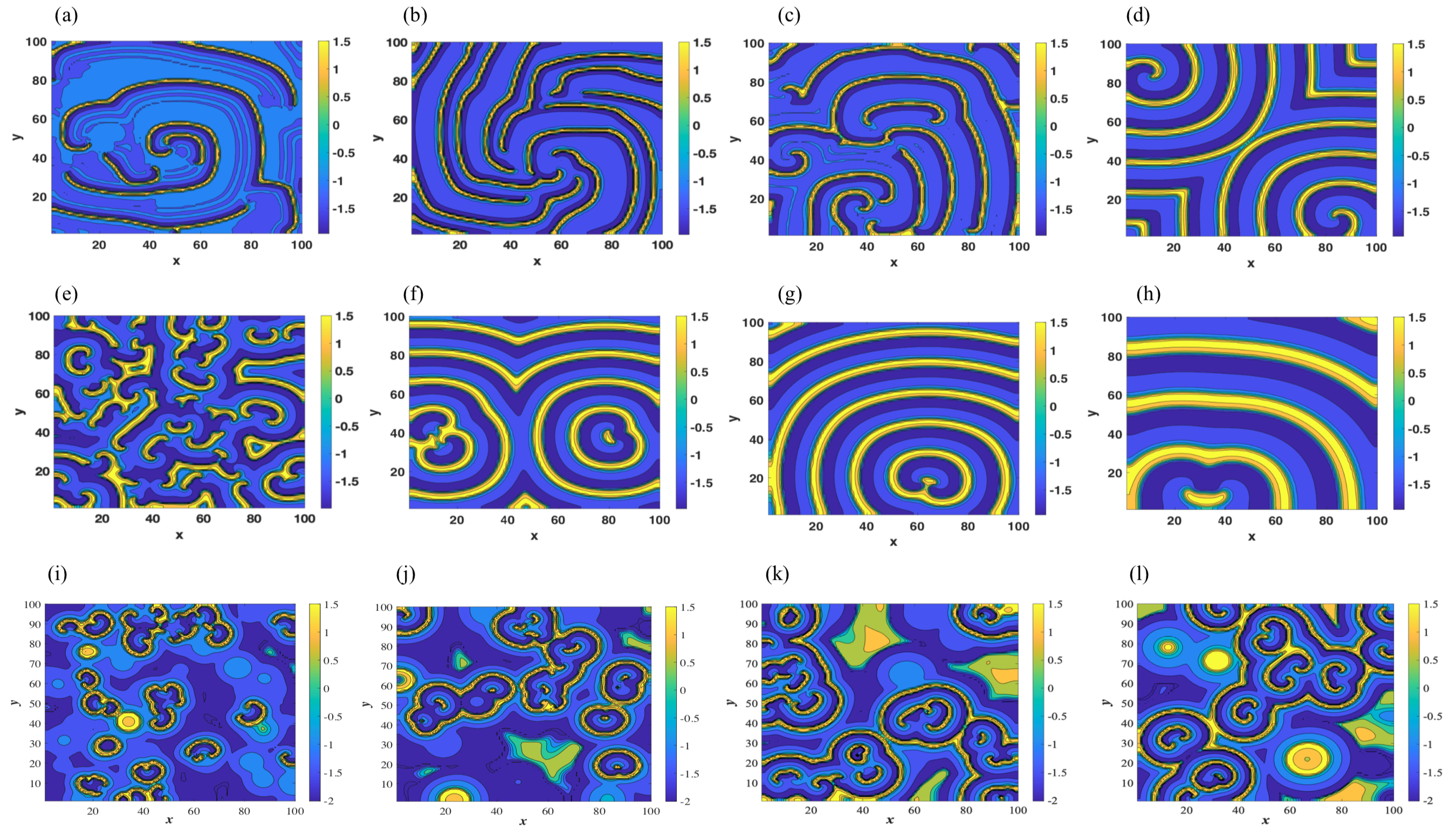}
\caption{Snapshots of pattern formation of the system of diffusively coupled FHR neurons \eqref{2-dimensionalmodel} arranged in a 2-dimensional configuration. Panel (a) $D=0.05$, (b) $D=0.09$, (c) $D=0.1$, (d) $D=0.245$ for set 1, panel (e) $D=0.1$, (f) $D=0.25$, (g) $D=0.3$, (h) $D=0.9$ for set 2, panel (i) $D=0.005$, (j) $D=0.007$, (k) $D=0.009$ and panel (l) $D=0.01$ for set 3. The color bars encode the values of the membrane voltages $u(x,y,t)$ of the diffusively coupled FHR neurons of system \eqref{2-dimensionalmodel}.}\label{patternD}
\end{figure*}

\section{Conclusions and Discussion}\label{Section5}

In this paper, we considered the biophysically motivated, slow-fast, excitable FitzHugh-Rinzel neuron model, which provides a variety of neuronal responses. It exhibits a diverse repertoire of firing patterns for certain fixed sets of parameters and different external current stimuli. We studied the FitzHugh-Rinzel model theoretically and numerically in different dynamical regimes. Further, we discussed its local and global stability. We performed a bifurcation analysis and demonstrated the qualitative changes between stable steady and oscillatory states. We incorporated 1-dimensional diffusion to construct a 1-dimensional chain of neurons and investigated several complex behaviors that pertain to synchronization. The stability analysis was performed for the diffusively coupled system and we showed how the stability changes from an unstable regime to a stable with the increase of the diffusion coefficient. The diffusively coupled system changes its dynamical characteristics as the coupling strength changes. Next, we studied the diffusive system on a square spatial domain. The amplitude equation, describing the onset of spirals close to the Hopf bifurcations, was derived and numerical simulations were provided. As the parameters change, we showed the appearance of target and spiral waves for three sets of parameters that correspond to three distinct oscillatory regimes.

We explored emerging target waves and spirals as dynamical features and verified them analytically using amplitude equations. Our investigation shows that the spatial dynamics of the slow-fast model exhibits two- and multi-arm spiral waves for low diffusion couplings. We found that multi-arm spiral waves with increasing number of arms occupy the 2-dimensional mesh grid and that they are unstable. Bursting in the activity of single neurons is a bio-physiological phenomenon, however, bursting in the activity of neural populations might be pathological \cite{jiang2015formation,meier2015bursting}. Interestingly, these emergent patterns may be relevant to the synchronized activities of neural populations, particularly related to neurological diseases \cite{meier2015bursting}. The propagation of neuronal impulses can be relevant to brain functioning \cite{grace2015reg,hu2013selection,keane2015propagating, wu2013formation}. In a weakly coupled system of pancreatic $\beta$-cells, bursting varies, where pancreatic $\beta$-cells secrete insulin \cite{raghavachari1999waves}. The analysis of the mechanisms underlying spatial profiles of activity in the neural tissue is important in understanding a wide range of biophysical and pathological phenomena \cite{huang2010spiral,izhikevich2007dynamical,jiang2015formation}. Moreover, the form of the nerve impulse propagation is relevant to certain brain pathologies \cite{huang2010spiral,schiff2007dynamical,song2018classification,townsend2018detection}.

Finally, we discussed in detail how nonlinearities in the biophysical, excitable model change the distribution of single-cell characteristics into different types of patterns. Spiral waves emerge frequently in cortical areas with limited lifespan. They can modify cortical activities affecting the oscillation frequency and spatial coherence-like activity. The emergence of spiral waves during sleep-like states varies greatly and, can also organize and modulate the pathological patterns during epilepsy. Spiral waves have been observed to play a major role in organizing irregular dynamics in cortical neurons and rhythmic behavior \cite{huang2010spiral}. Our results reveal a multitude of neural excitabilities and possible conditions for the emergence of spiral-wave formation in diffusively coupled FitzHugh-Rinzel systems with different firing characteristics.

\section*{Data availability}
Data sharing is not applicable to this article as no new data were created or analyzed in this study.

%\section*{ACKNOWLEDGMENTS}

% Argha Mondal (AM) is thankful for the support provided by the Department of Mathematical Sciences, University of Essex, UK and SKS is thankful to University Grants Commission (UGC), Govt. of India under NET-JRF scheme.

%\nocite{*}
%\bibliography{aipsamp}% Produces the bibliography via BibTeX.

%merlin.mbs aipnum4-1.bst 2010-07-25 4.21a (PWD, AO, DPC) hacked
%Control: key (0)
%Control: author (8) initials jnrlst
%Control: editor formatted (1) identically to author
%Control: production of article title (0) allowed
%Control: page (1) range
%Control: year (1) truncated
%Control: production of eprint (0) enabled
\bibliographystyle{unsrt}
%\bibliography{aipsamp.bib}

\end{document}